\documentclass[12pt]{article}
\usepackage[T1]{fontenc}
\usepackage[utf8]{inputenc}
\usepackage[english]{babel}
\pdfoutput=1

\oddsidemargin \evensidemargin
\usepackage[style=numeric]{biblatex}
\usepackage{microtype}
\usepackage{enumitem}
\setenumerate{label=\textup{(\roman*)}}

\usepackage{amsmath, mathtools}
\usepackage{amssymb}
\usepackage{amsthm, thmtools}
\usepackage{tikz-cd}
\usepackage{braket}
\usepackage{bm}
\usepackage{csquotes}
\usepackage{xspace}
\usepackage{float}

\usepackage[colorlinks]{hyperref}
\usepackage[nameinlink]{cleveref}
\hypersetup{
  linkcolor=[rgb]{0.3,0.3,0.6},
  citecolor=[rgb]{0.2,0.6,0.2},
  urlcolor =[rgb]{0.6,0.2,0.2}
}

\declaretheorem[numberwithin=section]{theorem}
\declaretheorem[numberlike=theorem]{lemma, corollary}
\declaretheorem[numberlike=theorem,style=definition]{definition}
\declaretheorem[numberlike=theorem,style=remark]{example}
\declaretheorem[numberlike=theorem,style=remark]{remark}
\declaretheoremstyle[
notefont=\bfseries, notebraces={(}{)}
]{conjecture}

\newcommand{\NN}{\mathbb{N}}
\newcommand{\ZZ}{\mathbb{Z}}
\newcommand{\QQ}{\mathbb{Q}}

\newcommand{\KK}{\mathbb{K}}
\newcommand{\FF}{\mathbb{F}}

\newcommand{\NULL}{\texttt{0}}
\newcommand{\EINS}{\texttt{1}}

\usepackage{tabularx, booktabs}

\newcommand\problemBox[4][0.96\textwidth]{
\begin{tabularx}{#1}{lX}
\toprule
\multicolumn{2}{c}{#2} \\ \midrule
\textbf{Input:}    & #3  \\
\textbf{Output:} & #4  \\ \bottomrule
\end{tabularx}
}

\newcommand{\Landau}{\mathcal{O}}

\newcommand{\pclass}[1]{\textsf{\textup{#1}}}
\newcommand{\Mon}{\pclass{Mon}}
\newcommand{\Homog}{\pclass{Homog}}

\newcommand{\x}{{\underline x}}
\renewcommand{\t}{{\underline t}}
\DeclareMathOperator{\ini}{in}
\DeclareMathOperator{\NF}{nf}

\newcommand{\problemText}[1]{\textsc{#1}\xspace}
\newcommand{\IdealMem}{\problemText{IdealMem}}
\newcommand{\AlgMem}{\problemText{AlgMem}}
\newcommand{\IniAlgMem}{\problemText{InAlgMem}}
\newcommand{\CSG}{\problemText{CSG}}

\newcommand{\OneInThreeSAT}{\problemText{1in3Sat}}
\newcommand{\SubsetSum}{\problemText{UnboundedSubsetSum}}
\newcommand{\LBAhalt}{\problemText{LBAhalt}}

\newcommand{\cclass}[1]{\textsf{\textup{#1}}\xspace}
\newcommand{\TCO}{\cclass{TC\ensuremath{^0}}}

\newcommand{\DP}{\cclass{P}}
\newcommand{\NP}{\cclass{NP}}
\newcommand{\PSPACE}{\cclass{PSPACE}}

\newcommand{\EXPSPACE}{\cclass{EXPSPACE}}

\newcommand{\leqlm}{\mathrel{\leq^{\cclass{L}}_{\textup{m}}}}

\newcommand{\ie}{i.\,e.\@\xspace}
\newcommand{\eg}{e.\,g.\@\xspace}

\title{Computational Complexity of Polynomial Subalgebras}
\author{Leonie Kayser}
\date{}

\begin{document}
\maketitle

\begin{abstract}
The computational complexity of polynomial ideals and Gröbner bases has been studied since the 1980s. In recent years, the related notions of polynomial subalgebras and SAGBI bases have gained more and more attention in computational algebra, with a view towards effective algorithms. We investigate the computational complexity of the subalgebra membership problem and degree bounds. In particular, we show completeness for the complexity class \EXPSPACE and prove \PSPACE-completeness for homogeneous algebras. We highlight parallels and differences compared to the settings of ideals, and also look at important classes of polynomials such as monomial algebras.
\end{abstract}

\section{Introduction}

Let $\KK$ be a field equipped with a suitable encoding over a finite alphabet, for example, the rational numbers $\QQ$ or a finite field $\FF_q$.
At the heart of many algorithms in computer algebra lie efficient algorithms manipulating polynomial ideals and Gröbner bases, dating back at least to the 1960s \cite{Buchberger2006}. Since the late 1980s, algorithms for computations with subalgebras have been studied \cite{Shannon1988, Robbiano1990}, which have applications in toric degenerations and polynomial system solving \cite{betti2025}. A fundamental computational problem for ideals and subalgebras of polynomial rings $\KK[\x] = \KK[x_1,\dots,x_n]$ is deciding \emph{membership}. Of special interest are classes $\textsf{C}$ of polynomials such as homogeneous polynomials (\Homog), monomials (\Mon), or polynomials with bounded degree or number of variables.

\begin{table}[H]
\centering
\problemBox[0.48\textwidth]{$\IdealMem_\KK(\textsf{C})$, $\IdealMem_\KK$}
{$f_1,\dots,f_s \in \textsf{C}$ (or $\KK[\x]$);\newline $g \in \KK[\x]$}
{Is $g \in \langle f_1,\dots,f_s \rangle$?}
\problemBox[0.48\textwidth]{$\AlgMem_\KK(\textsf{C})$, $\AlgMem_\KK$}
{$f_1,\dots,f_s \in \textsf{C}$ (or $\KK[\x]$);\newline $g \in \KK[\x]$}
{Is $g \in \KK[f_1,\dots,f_s]$?}
\end{table}

Recall that $g \in \langle f_1,\dots,f_s \rangle$ if and only if there is a representation $g = h_1f_1+\dots+ h_sf_s$ with $h_1,\dots,h_s \in \KK[\x]$; the \emph{representation degree} (of the tuple $f_1,\dots,f_s,g$) is the smallest possible value of $D = \max \{\deg h_1,\dots,\deg h_s\}$.
Similarly, $g \in \KK[f_1,\dots,f_s]$ if and only if $g = p(f_1,\dots,f_s)$ for some \emph{certificate} $p \in \KK[t_1,\dots,t_s]$; we call the smallest possible degree of $p$ the \emph{certification degree}.

The computational complexity and representation degree bounds of $\IdealMem_\KK$ and its variants have been studied since the 1980s, we give a brief overview in \Cref{sec:ideal_world}. The complexity of algebras in a more general sense has been studied before \cite{Kozen77}, though so far not for the concrete case of polynomial subalgebras. This paper is concerned with filling this gap and proving analogous results for $\AlgMem_\KK$ and bounds on the certification degree. In order to obtain reasonable upper bounds on the computational complexity, we assume that arithmetic in the base field $\KK$ can be performed efficiently; this is formalized as being \emph{well-endowed} \cite{Borodin1983}. A main theorem is the \EXPSPACE- and \PSPACE-completeness of the subalgebra membership problem and the homogeneous variant, combining the upper bounds \Cref{thm:algmem_expspace} and \ref{thm:algmem_homog} and the complementing lower bound \Cref{thm:reduction_ideal_algebra}.

\begin{theorem}\label{thm:main_result}
Over a well-endowed field $\KK$, the subalgebra membership problem $\emph{\AlgMem}_\KK$ is \EXPSPACE-complete and the homogeneous version $\emph{\AlgMem}_\KK(\Homog)$ is \PSPACE-complete.
\end{theorem}

One of the goals of this paper is to be accessible both to commutative algebraists and computer scientists. For this reason, we introduce the relevant results from algebra and complexity theory on ideals and Gröbner bases in \Cref{sec:background}. In \Cref{sec:upper_bounds} we prove various upper bound results on variations of $\AlgMem_\KK$, as well as an upper bound on the certification degree. In \Cref{sec:lower_bounds} we present a complexity-theoretic reduction from $\IdealMem_\KK$ to $\AlgMem_\KK$, proving matching lower bounds. In \Cref{sec:monomials} we study the case of monomial subalgebras, which is still \NP-complete, and discuss related questions with SAGBI bases.

\section{Background in algebra and complexity theory}\label{sec:background}

\subsection{Ideals, subalgebras and their bases}\label{sec:ideals}

Let $\KK$ be a field, denote by $\KK[\x] \coloneqq \KK[x_1,\dots,x_n]$ the polynomial ring in $n$ variables, consisting of finite linear combinations
\[
f = \sum_{\alpha \in \NN^n} c_\alpha  x^\alpha, \qquad x^\alpha = x_1^{\alpha_1}\dotsm x_n^{\alpha_n}, \quad c_\alpha \in \KK.
\]
Ideals and subalgebras arise naturally in algebra as the kernels and images of ($\KK$-linear) ring homomorphisms. We briefly recall the basic definitions for both of them in parallel.

\begin{definition}
A subset $I \subseteq \KK[\x]$ is an \emph{ideal} if it is an additive subgroup (contains $0$ and is closed under addition and subtraction) and also closed under multiplication with \emph{arbitrary} polynomials $f \in \KK[\x]$. A subset $A \subseteq \KK[\x]$ is a \emph{subalgebra} if it is an additive subgroup containing $\KK$ and is closed under multiplication within $A$.
Given a set $S \subseteq \KK[\x]$, we denote by
\begin{align*}
  \langle S \rangle &= \Set{ \sum_{\text{finite}} h_i f_i | \!\!\!\begin{array}{c}
f_i \in S,\\ h_i \in \KK[\x]
  \end{array}\!\!\!} = \set{\text{lin.\ comb.\ of $S$ over $\KK[\x]$}} \\
  \KK[S] &= \Set{ \sum_{\text{finite}} c_\alpha f^\alpha | \!\!\!\begin{array}{c}
   f_i \in S,\\ c_\alpha \in \KK
  \end{array}\!\!\!  } = \set{\text{polynomial expressions in $S$}}
\end{align*}
the ideal and subalgebra generated by $S$; it is the smallest ideal resp.\ subalgebra containing $S$.
\end{definition}

\begin{example}
An important class of ideals and algebras are those generated by monomials $x^\alpha$, or binomials $x^\alpha-x^\beta$. For example, monomial algebras show up as coordinate rings of affine toric varieties, and binomial ideals are related to commutative Thue systems, capturing the combinatorial worst-case complexity of ideal membership (see \Cref{thm:idealmem_expspace}).
\end{example}

\begin{definition}
A \emph{monomial order} $\prec$ is a total order on the set of monomials $x^\alpha$ stable under multiplication ($x^\alpha \preceq x^\beta$ implies $x^\alpha x^\gamma \preceq x^\beta x^\gamma$) and such that $1 = x^0$ is the minimal element.
The \emph{initial term} $\ini_\prec(f)$ of a nonzero polynomial $f = \sum_\alpha c_\alpha x^\alpha$ is the nonzero term $c_\alpha x^\alpha$ with the largest monomial. For convenience define $\ini_\prec(0) = 0$. Given an ideal $I$ or a subalgebra $A$, its initial ideal or initial algebra, respectively, is
\[
\ini_\prec(I) = \langle \set{\ini_\prec(f) | f \in I} \rangle, \quad \ini_\prec(A) = \KK[\set{\ini_\prec(f) | f \in A}].
\]
\end{definition}

In the following fix such a monomial order, for example, the \emph{lexicographic order} $\prec_{\text{lex}}$; we may assume $x_1>\dots>x_n$.

\begin{definition}\label{def:Groebner_SAGBI}
A \emph{Gröbner basis} of an ideal $I$ is a \emph{finite} subset $G \subseteq I$ such that $\ini_\prec(I) = \langle \set{\ini_\prec(g) | g \in G}\rangle$. A Gröbner basis is \emph{reduced} if the leading coefficients are $1$ and no term of $g \in G$ is in $\ini_\prec(G\setminus\{g\})$.

On the other hand, a \emph{SAGBI basis} of a subalgebra $A$ is \emph{any} subset $S \subseteq A$ such that $\ini_\prec(A) = \KK[\set{\ini_\prec (f) | f \in S}]$.
\end{definition}
In both cases, such bases generate $I$ and $A$ respectively, hence the (historical) name \enquote{basis}. The reduced Gröbner basis of an ideal $I$ is unique and has the smallest number of elements and smallest degree among all Gröbner bases of $I$. For more background on Gröbner bases see for example \cite[Chapter 2]{Cox2016}, \cite[Chapter 21]{Gathen2017} or \cite[Chapter 2]{Kreuzer2000}.

By \emph{Hilbert's basis theorem}, every ideal can be generated by a finite set of polynomials; in fact one can find such a finite set in any generating set. This implies the existence of (finite) Gröbner bases for any ideal. On the other hand, finite SAGBI bases may not exist; we will come back to this topic in \Cref{sec:SAGBIconjecture}.

\begin{definition}
Given an ideal $I \subseteq \KK[\x]$, the \emph{normal form} $\NF_\prec(f,I)$ of an element $f \in \KK[\x]$ against $I$ is the unique polynomial $r \in f+I$ such that no monomial of $r$ is in $\ini_\prec(I)$ \cite[Definition 2.4.8]{Kreuzer2000}.
\end{definition}

One can compute $\NF_\prec(f,I)$ by dividing $f$ by a Gröbner basis $G$, the remainder is the normal form \cite[Proposition 2.6.1]{Cox2016}. This gives a membership test for polynomial ideals: $f \in I$ if and only if $\NF_\prec(f,I) = 0$. 

\subsection{Complexity theory}\label{sec:complexity}

We use the Turing model of computation, though the details of this formalism are not required to follow this paper. Here we introduce the concepts and computational problems used in the sequel.
A classical reference for computational complexity theory is the book by Hopcroft and Ullman \cite{Hopcroft1979}, a more modern treatment is the book by Arora and Barak \cite{Arora2009}, for a brief introduction with a view towards computer algebra, see also \cite[Section 25.8]{Gathen2017}.

Informally, the complexity of an algorithm is the amount of resources (time or memory) used in the computation as a function of the \emph{input length}. 

\begin{definition}
A \emph{decision problem} $A$ is the problem of deciding whether a given input $x \in \Sigma^* \coloneqq \{\text{words over the alphabet }\Sigma\}$ has a certain property, \ie deciding membership in the set $A \subseteq \Sigma^*$.
Problems in \DP, \PSPACE, and \EXPSPACE can be solved algorithmically in polynomial time, polynomial space, and exponential space, respectively, while the answer to problems in \NP can be \emph{verified} in polynomial time provided a certificate.
\end{definition}
One has the inclusions of complexity classes
\[
\DP \subseteq \NP \subseteq \PSPACE \subsetneq \EXPSPACE.
\]
A standard problem in \NP is \problemText{Sat}, deciding whether a given boolean formula $\varphi(x_1,\dots,x_n)$ can be satisfied by some assignment $\{x_1,\dots,x_n\} \to \{\texttt{true},\texttt{false}\}$. We will later need the useful variant \OneInThreeSAT (compare \cite{Schaefer78}, also known as \emph{positive/monotone} \OneInThreeSAT) and the \SubsetSum problem:

  \begin{table}[H]
    \centering
    \problemBox[0.48\textwidth]{\OneInThreeSAT}
    {$S_1,\dots,S_n \subseteq \NN$, $|S_i| \leq 3$}
    {Is there a set $T \subseteq \NN$ such that $|T \cap S_i| = 1$ for all $i$?} 
    \problemBox[0.48\textwidth]{\SubsetSum}
    {$a_1,\dots,a_s;b \in \NN$}
    {Is $\sum_{i = 1}^s c_ia_i = b$ for some $c_1,\dots,c_s \in \NN$?}
  \end{table}

A typical example of a problem in \PSPACE is the halting problem for (deterministic) linearly bounded automata (LBA). Here a LBA $M$ consists of a finite set of states $Q$ including a starting state $q_0$ and a halting state $q_{\text{halt}}$, a tape alphabet $\Gamma = \{\NULL,\EINS,{\triangleright},{\triangleleft}\}$ containing the input alphabet $\Sigma = \{\NULL,\EINS\}$, and a transition function
\[
\delta \colon (Q\setminus\{q_{\text{halt}}\}) \times \Gamma \to Q \times \Gamma \times \{{\rm L},{\rm R}\}.
\]
A configuration is a triple $(q,i,b_0b_1\dots) \in Q \times \ZZ \times \Gamma^*$, and if $\delta(q,b_i) = (q',c,X)$, then $M$ transitions as
\[
(q,i,\dots b_{i-1}b_ib_{i+1} \dots ) \Rightarrow \begin{cases}
  (q',i-1, \dots b_{i-1}c b_{i+1} \dots) & \text{if }X = {\rm L},\\
  (q',i+1, \dots b_{i-1}c b_{i+1} \dots) & \text{if }X = {\rm R}. 
\end{cases} 
\]
On the input string $w \in \Sigma^*$, the machine $M$ starts on the starting configuration $(q_0,1,{\triangleright}w_1\dots w_n {\triangleleft})$ (here $\triangleright$ is at position $0$) and repeatedly transitions according to $\delta$. The end markers $\triangleright, \triangleleft$ may never be overwritten and the head may not pass over them (e.g. $\delta(q,{\triangleleft}) = (q',{\triangleleft},{\rm R})$ is not allowed). $M$ halts if it eventually reaches a configuration of the form $(q_{\text{halt}},i,b)$.

\begin{table}[H]
  \centering
  \problemBox[\textwidth]{\LBAhalt}
  {A LBA $M = (Q,\delta,q_0,q_{\text{halt}})$ and an input $w \in \{\NULL,\EINS\}^*$}
  {Does $M$ halt on input $w$?}
\end{table}

\begin{example}\label{exa:binaryLBA}
  The following table shows the transition function of a LBA consists of two states $q_0,q_1$ (plus a halting state), and on input $w \coloneqq \NULL\dots \NULL$ ($n$ zeros) will count in binary to $\EINS \dots \EINS$ and then halt:
  \[
  \begin{array}{c|cccc}
  b & {\triangleright}               & \NULL               & \EINS               & {\triangleleft} \\ \hline \hline
  q_0 &                                & (q_1,\EINS,{\rm L}) & (q_0,\NULL,{\rm R}) & \text{halt}      \\
  q_1 & (q_0,{\triangleright},{\rm R}) & (q_1,\NULL,{\rm L}) &                     &
  \end{array}
  \]
Transitions that do not occur (on input $w$) have been left unspecified for simplicity. An example of the transition sequence can be found at
\[
\text{\href{https://mathrepo.mis.mpg.de/ComplexityOfSubalgebras/configurations.html}{\small \texttt{mathrepo.mis.mpg.de/ComplexityOfSubalgebras/configurations.html}}}
\]
\end{example}

\begin{definition}
A \emph{complexity upper bound} for a decision problem $A$ and a complexity class $\cclass C$ is an algorithm solving $A$ and satisfying the resource constraints of $\cclass C$.
\end{definition}
A \emph{lower bound} or $\cclass C$\emph{-hardness} result for $A$ is a proof that the computation of any problem $A' \in \cclass C$ can be reduced to the problem $A$. We use the notion of log-space many-one reductions:
\begin{definition}
For two sets $A\subseteq \Sigma^*$, $A' \subseteq (\Sigma')^*$, the problem $A'$ is \emph{log-space many-one reducible} to $A$, in symbols $A' \leqlm A$, if there is a map $f\colon (\Sigma')^*\to \Sigma^*$ computable in logarithmic working space, such that $x \in A'$ if and only if $f(x) \in A$.

A problem $A$ is \emph{$\cclass C$-hard}, if $A' \leqlm A$ for all $A'$ in $\cclass C$. It is \emph{$\cclass C$-complete} if it is also in $\cclass C$.
\end{definition}
For example, \problemText{Sat}, \OneInThreeSAT, and \SubsetSum are \NP-complete, and \LBAhalt is \PSPACE-complete.

In the problems $\IdealMem_\KK$ and $\AlgMem_\KK$, the input consists of encoded polynomials over $\KK$, for example as a string of characters. Choosing a reasonable representation, the input length is bounded below by the number of terms, the number of variables occurring, and, depending on the encoding, the coefficients and the (logarithm of the) degree.

\begin{definition}
The encoding of a polynomial $f = \sum_{|\alpha|\leq d} c_\alpha x^\alpha$ is \emph{dense} if it lists all monomials with their coefficients $(\alpha,c_\alpha)$ up to its degree, and \emph{sparse} if it only lists those terms with non-zero coeffient. The representation of a monomial $x^\alpha \mathbin{\hat=} (\alpha_1,\dots,\alpha_n) \in \NN^n$ is in \emph{binary} if each $\alpha_i$ is encoded in binary, and \emph{unary} if it is encoded in unary.
\end{definition}
\begin{example}
The term $f=42x_1^{3}x_3^6 \in \QQ[x_1,x_2,x_3]$ could be encoded with binary exponents as $((11,0,110),\texttt{enc}(42))$, and with unary as $((111,0,111111),\texttt{enc}(42))$. In dense encoding, a representation of $f$ would have to list all $\binom{9+3}{3} = 220$ monomials up to degree $9$ (which most have coefficient $0$), while a sparse encoding only has to list that one term.
\end{example}

All our complexity results will hold with respect to both dense and sparse encodings and binary or unary monomial representation (with the exception of \Cref{thm:Monomial_bounded}), as long as we make the reasonable assumption that the field $\KK$ is \emph{well-endowed} \cite{Borodin1983}. Commonly used fields in computer algebra such as number fields and finite fields are well-endowed. For this reason, we will not elaborate further on the encoding.

\subsection{The ideal world: brief summary of results}\label{sec:ideal_world}

In this section, we give an overview of the complexity results regarding ideal membership, representation degree, and Gröbner basis degrees. Our upper bounds on $\AlgMem_\KK$ will rely on refined representation and Gröbner basis degree upper bounds presented below.
A more comprehensive discussion of the complexity of ideals can be found in \cite{Mayr2017}.

\begin{theorem}\label{thm:idealmem_expspace}
Let $\KK$ be a well-endowed field (the lower bounds do not require this).
\begin{enumerate}
\item \textup{(Mayr \& Mayer \cite{Mayr1982}, Mayr \cite{Mayr1989})} The problem $\emph{\IdealMem}_\KK$ is complete for \EXPSPACE. A representation can be enumerated in exponential working space.
\item \textup{(Mayr \cite{Mayr1997})} The problem $\emph{\IdealMem}_\KK(\Homog)$ is complete for \PSPACE. The general problem is also in $\PSPACE$ when bounding the number of variables.
\end{enumerate}
All hardness results already hold for ideals generated by binomials $x^\alpha-x^\beta$.
\end{theorem}

\begin{remark}\label{rmk:double_exponential}
Note that this uses the convention that in a computational problem with output, we only consider the working space and not the space needed to store the entire output. For example, one can enumerate the binary numbers $0,1,\dots,2^n-1$ with only $\Landau(n)$ bits of working memory. Similarly, here the number of terms in a representation $(h_i)_i$ may be doubly exponential in the input length; this is unavoidable in the worst case \cite{Mayr1982}.
\end{remark}

Representation degree bounds are, in general, doubly exponential in the number of variables, which is asymptotically optimal. We also note the special case of complete intersections for later reference. For a definition of complete intersections see \eg \cite[Exercise 9.4.8]{Cox2016} or \cite[Definition 3.2.23]{Kreuzer2000}.

\begin{theorem}\label{thm:hermann_bound}
Let $g \in \langle f_1,\dots,f_s \rangle \subseteq \KK[x_1,\dots,x_n]$, $d \coloneqq \max_i \deg f_i$, and write $g = \sum_{i=1}^s h_if_i$ with $h_i \in K[\x]$ of minimal representation degree $D \coloneqq \max_i \deg h_i$.
\begin{enumerate}
\item \textup{(Hermann \cite{Hermann1926}, Mayr \& Meyer \cite{Mayr1982})} $D \leq \deg g + (ds)^{2^n}$.
\item \textup{(Dickenstein, Fitchas, Giusti \& Sessa \cite{Dickenstein1991})} If $f_1,\dots,f_s$ define a complete intersection ideal, then $D \leq \deg g + d^s$.
\end{enumerate}
\end{theorem}

Both from a theoretical and computational point of view, effective degree bounds on \emph{Gröbner bases} are important as well. We recall the classical \emph{Dubé bound} as well as a dimension-refined version. Here, the dimension of an ideal $I \subseteq \KK[\x]$ is the (Krull) dimension of the quotient ring $\KK[\x]/I$, see \cite[Chapter 9]{Cox2016} for other characterizations of the dimension.

\begin{theorem}\label{thm:gb_dim_mayr_ritscher}
Let $I \subseteq \KK[x_1,\dots,x_n]$ be an ideal generated by $f_1,\dots,f_s$ of degree $d_i \coloneqq \deg f_i$, $d_1\geq \dots \geq d_s$, and consider any monomial order.
\begin{enumerate}
\item \textup{(Dubé \cite{Dube1990})} The degree of the reduced Gröbner basis $\textit{GB}$ of $I$, \ie the largest degree of an element of $\textit{GB}$, is bounded by
\[
\deg \textit{GB} \leq 2 \left( \frac{1}{2}d_1^2 + d_1\right)^{2^{n-1}}.
\]
\item \textup{(Mayr \& Ritscher \cite{Mayr2013})} If $I$ has dimension $r$, then the degree is bounded by
\[
\deg \textit{GB} \leq 2 \left( \frac{1}{2} (d_1\dotsm d_{n-r})^{2(n-r)} + d_1\right)^{2^r} \! \leq 2 \left( \frac{1}{2} d_1^{2(n-r)^2} + d_1\right)^{2^r}.
\]
\end{enumerate}
These upper bounds are asymptotically sharp.
\end{theorem}

\section{Upper bounds on subalgebra membership}\label{sec:upper_bounds}

In this section we provide various complexity upper bounds, that is, algorithms for variants of $\AlgMem_\KK$ which place these problems in \PSPACE and \EXPSPACE.
Our complexity analysis of the membership problem for a subalgebra $A = \KK[f_1,\dots,f_s] \subseteq \KK[\x]$ relies on the following classical elimination approach using tag variables, attributed to Spear \cite{Shannon1988, Spear1977}.

Let $\t = \{t_1,\dots,t_s\}$ be additional variables, and consider an \emph{elimination order} on $\KK[\x,\t]$, i.\,e.\ a monomial order such that $x_i \succ t^\alpha$ for all $x_i$ and $t^\alpha \in \KK[\t]$, for example the lexicographical order $\prec_{\text{lex}}$.

\begin{lemma}[Shannon \& Sweedler {\cite{Shannon1988}}]\label{thm:algmem_elimination}
Let $f_1,\dots,f_s,g \in \KK[\x]$, $J\coloneqq \langle f_1-t_1,\dots,f_s-t_s\rangle$, then
\[
g \in \KK[f_1,\dots,f_s] \qquad \text{if and only if} \qquad p \coloneqq \NF_\prec(g,J) \in \KK[\t].
\]
If this is the case, then $g = p(f_1,\dots,f_s)$, interpreting $p$ as a polynomial in $t_1,\dots,t_s$.
\end{lemma}

Subalgebra membership can thus be reduced to the task of calculating the normal form of a polynomial against an ideal in a larger polynomial ring. This approach is well-known, for example as a fall-back method in the \texttt{Macaulay2} \cite{M2} package \texttt{SubalgebraBases} \cite{Burr2024}.
Kühnle \& Mayr \cite{Kühnle1996} describe an exponential space algorithm for enumerating the normal form of a polynomial over a well-endowed field. Combining these results gives the first complexity upper bound:

\begin{theorem}\label{thm:algmem_expspace}
For any well-endowed field $\KK$, $\emph{\AlgMem}_\KK$ is in \EXPSPACE. Moreover a certificate $p \in \KK[t_1,\dots,t_s]$ such that $g = p(f_1,\dots,f_s)$ can be output in exponential working space.
\end{theorem}
\begin{proof}
Given $f_1,\dots,f_s,g$, the algorithm computes the normal form $p \coloneqq \NF_\prec(g,J) \in \KK[\x,\t]$ using Kühnle \& Mayr's algorithm \cite{Kühnle1996}. If there is a nonzero term in $p$ involving some $x_i$, then $g \notin \KK[f_1,\dots,f_s]$. Otherwise, $p$ is a certificate of subalgebra membership by \Cref{thm:algmem_elimination}, which can be enumerated to the output tape using single exponential working space.

Note that it might not be possible to fit the normal form in exponential working space as mentioned in \Cref{rmk:double_exponential}. Instead one has to enumerate the terms of $p$, which fit in exponential working space using and check for occurrences of $x_i$ individually.
\end{proof}

To prove degree bounds and better complexity upper bounds for special classes of polynomials, we outline part of Kühnle \& Mayr's upper bound constructions on normal forms \cite{Kühnle1996}.

Let $I \subseteq \KK[x_1,\dots,x_n]$ be an ideal, $g \in \KK[\x]$ and ${\prec}\coloneqq {\prec_{\text{lex}}}$ the lexicographic order (for simplicity). Let $\mathit{GB} = \mathit{GB}(I)$ be the reduced Gröbner basis of $I$, then 
\[
\deg \NF_{\prec}(g,I) \leq \deg(g) + (\deg(\mathit{GB})+1)^{n^2+1} \deg(g)^n.
\]
Let $D$ be an upper bound on the degree of the coefficients $h_i \in K[\x]$ in a representation $g-\NF_{\prec}(g,I) = \sum_{i=1}^s h_if_i$, for example the Hermann bound from \Cref{thm:hermann_bound}, then Kühnle \& Mayr reduce the normal form calculation into a linear algebra problem of size $\operatorname{poly}(D) = D^{\Landau(1)}$. Using efficient parallel algorithms for matrix rank and the parallel computation hypothesis, this yields an algorithm in space $\operatorname{polylog}(D) = (\log D)^{\Landau (1)}$. Using the Dubé bound for $\deg(\mathit{GB})$ and the Hermann bound for $D$ yields the mentioned exponential space algorithm for normal form calculation.

In our specialized setting $J = \langle t_1-f_1,\dots,t_s-f_s \rangle$ we have more refined upper bounds available regarding representation and Gröbner basis degrees. We apply this to the case of a fixed number of variables $n$ and to the case of a fixed subalgebra $A$.

\begin{theorem}\label{thm:algmem_bounded_vars}
For a fixed subalgebra $A = \KK[f_1,\dots,f_s] \subseteq \KK[\x]$ over a well-endowed field, the membership problem is in \PSPACE (with respect to the input length of $g$). In fact, for a \emph{fixed} number of variables $n$, we have $\emph{\AlgMem}_\KK(\KK[x_1,\dots,x_n]) \in \PSPACE$.
\end{theorem}
\begin{proof}
The algorithm is the same as in \Cref{thm:algmem_expspace}, only the complexity analysis is slightly more elaborate. Using that normal form calculation with respect to a fixed ideal is possible in polynomial space, as remarked by Kühnle \& Mayr \cite[Section 6]{Kühnle1996}, the first assertion follows.

If only the number of variables $x_1,\dots,x_n$ is fixed, then more care has to be taken, as the computation takes place in the ring $\KK[x_1,\dots,x_n,t_1,\dots,t_s]$, whose number of variables $n+s$ still depends on the input length. However, the ideal $J = \langle t_1-f_1,\dots,t_s-f_s\rangle$ is a complete intersection of dimension $n$, hence using \Cref{thm:gb_dim_mayr_ritscher} its Gröbner basis degree is bounded above by
\[
\deg \textit{GB}(J) \leq G \coloneqq 2 \left( \frac{1}{2} d^{2s^2} + d\right)^{2^n}.
\]
Furthermore, the representation degree of $g-\NF_\prec(g,J)$ is bounded above by $D \coloneqq G + d^s$ by \Cref{thm:hermann_bound}. Since $n$ is fixed, we see that $D$ is single exponential in the input length, hence $\operatorname{polylog}(D)$ is polynomial. This shows that the Meyer \& Kühnle algorithm works in polynomial space for a bounded number of variables.
\end{proof}

\begin{corollary}\label{cor:algmem_degree_bound}
If $g \in \KK[f_1,\dots,f_s] \subseteq \KK[\x]$, then there exists a certificate $p \in \KK[t_1,\dots,t_s]$ with $g = p(f_1,\dots,f_s)$ of degree
\[
\deg p \leq \deg(g) + \left(\left( \textstyle\frac{1}{2} d^{2s^2} + d\right)^{2^n}+1\right)^{(n+s)^2+1} \deg(g)^{n+s}.
\]
\end{corollary}
\begin{proof}
This is the degree bound on the normal form by Kühnle \& Mayr evaluated for the ideal $J$ (note that the ambient polynomial ring has $n+s$ variables).
\end{proof}
\begin{remark}
This degree can probably be improved by studying the derivation of the normal form degree bounds in this particular situation. We suspect that a \enquote{nicer} upper bound akin to the Hermann bound (\Cref{thm:hermann_bound}) should hold.
\end{remark}

Of great interest is the case of subalgebras generated by homogeneous polynomials. In this case, $A$ is graded compatibly with the standard grading of the ambient ring.

\begin{theorem}\label{thm:algmem_homog}
For any well-endowed field $\KK$, $\emph{\AlgMem}_\KK(\Homog)$ is in \PSPACE.
\end{theorem}
\begin{proof}
As a consequence of the grading, we see that a certificate $p$ of lowest degree (if it exists) has degree at most $\deg g$. Thus either $\deg \NF_\prec(g,J) > \deg g$, in which case $g \notin K[f_1,\dots,f_s]$, or $\deg \NF_\prec(g,J) \leq \deg g$ and $g \in \KK[f_1,\dots,f_s]$ if and only if $\NF_\prec(g,J) \in K[\t]$. This leads to a variation of the algorithm of \Cref{thm:algmem_expspace} except we only attempt to compute the normal form up to degree $\deg g$. The complexity analysis is analogous to \Cref{thm:algmem_bounded_vars}, noting that $\deg \NF_\prec(g,J)\leq \deg g$, hence a representation degree bound here is $D \coloneqq \deg g + d^s$.
\end{proof}

\begin{remark}
Degree bounds on certificates for elimination problems have been studied before, for example by Galligo \& Jelonek \cite{Galligo2020elimination}. Their results are not necessarily applicable here, as our ideal $J$ has dimension $n$ in a ring of $n+s$ variables, while the results of Galligo \& Jelonek only apply to intersections of an ideal with $K[x_1,\dots,x_n]$ if the ideal has dimension at least $n+1$.
\end{remark}

\section{Lower bounds on subalgebra membership}\label{sec:lower_bounds}

In this section, we prove matching complexity lower bounds for the subalgebra membership problem and its homogeneous variant, relating them to the construction for homogeneous ideals.
The hardness results for $\IdealMem_\KK$ and $\IdealMem_\KK(\Homog)$ were constructed by embedding the word problem for commutative semigroups resp.\ its homogeneous variant into binomial ideals \cite{Mayr1982, Mayr1997}. We now describe how to embed (binomial) ideal membership into subalgebra membership.

\begin{theorem}\label{thm:ideal_to_algebra}
Let $f_1,\dots,f_s,g \in \KK[x_1,\dots,x_n]$ and let $t$ be a new variable. The following are equivalent:
\begin{enumerate}[label=\textup{(\alph*)}]
\item $g \in \langle f_1,\dots,f_s\rangle \subseteq \KK[\x]$;
\item $tg \in \KK[tf_1,\dots,tf_s,x_1,\dots,x_s]$.
\end{enumerate}
\end{theorem}
\begin{proof}
If $g = h_1f_1+\dots+h_sf_s$ with $h_i \in \KK[\x]$, then
\[
tg = h_1\cdot (tf_1)+\dots+h_s\cdot (tf_s),
\]
which certifies that $tg$ is in $\KK[tf_1,\dots,tf_s,\x]$.

Conversely, assume $tg = p(tf_1,\dots,tf_s,x_1,\dots,x_n)$ with $p \in R \coloneqq \KK[u_1,\dots,u_s,x_1,\dots,x_n]$. Declare a (non-standard) grading on $R$ and $S=\KK[t,x_1,\dots,x_n]$ by giving $\deg t = \deg u_j = 1$ and $\deg x_i = 0$, so
\[
R_d \coloneqq \bigoplus_{\substack{\alpha \in \NN^s\\ |\alpha|=d}} u_1^{\alpha_1}\dotsm u_s^{\alpha_s}\KK[x_1,\dots,x_n], \qquad S_d \coloneqq t^d\KK[x_1,\dots,x_n].
\]
The evaluation homomorphism $R\to S$ substituting $x_i \mapsto x_i$, $u_j \mapsto t_jf_j$ respects this grading. Since $tg \in S_1$ is homogeneous of degree $1$, decomposing $p$ into its graded components $p=\sum_d p_d$ ($p_d \in R_d$) reveals that
\[
p_d(tf_1,\dots,tf_s,x_1,\dots,x_n) = \begin{cases}
p(tf_1,\dots,x_n) = tg & \text{if }d=1 \\
0 & \text{if }d\neq 1.
\end{cases}
\]
Thus by replacing $p$ by $p_1$, we may assume that every term in $p$ involves exactly one $u_j$. Setting $t=1$, we read
\[
p(f_1,\dots,f_s,x_1,\dots,x_n) = g,
\]
which, by the structure of $p$, is a representation for $g \in \langle f_1,\dots,f_s\rangle$.
\end{proof}
\begin{corollary}\label{cor:repdeg_to_certdeg}
Follow the notation from the previous theorem and set $A \coloneqq \KK[tf_1,\dots,tf_s,\x]$.
\begin{enumerate}
\item The representation degree of $g \in \langle f_1,\dots,f_s\rangle$ is one less than the certification degree of $tg \in A$.
\item The smallest possible number of non-zero terms of all $h_i$ in a representation $g = h_1f_1+\dots+h_sf_s$ equals the smallest possible number of terms of a certificate for $tg \in A$. 
\end{enumerate}
\end{corollary}
\begin{proof}
The forward direction in the previous proof gives that the total degree and number of terms for representation of ideal membership are a lower bound for the respective values for the subalgebra (the total degree increases by $1$ by attaching the $u_j$ to the $h_j$). Conversely, the grading argument shows that if $p$ is a certificate, then so is its graded component $p_1$, which has a lower or equal degree and number of terms. Finally, setting $t\mapsto 1$ in $p_1$ gives a representation for $g \in \langle f_1,\dots,f_s\rangle$, showing the desired equalities.
\end{proof}

With this result we can deduce matching complexity lower bounds for subalgebra membership and its homogeneous variant. We can rephrase \Cref{thm:ideal_to_algebra} as a complexity-theoretic reduction.

\begin{theorem}\label{thm:reduction_ideal_algebra}
The map
\[
(f_1,\dots,f_s;g \in \KK[\x]) \;\mapsto\; (tf_1,\dots,tf_s,x_1,\dots,x_n;tg\in\KK[t,\x])
\]
provides a reduction $\emph{\IdealMem}_\KK \leqlm \emph{\AlgMem}_\KK$. More generally, it provides a reduction $\emph{\IdealMem}_\KK(\textup{\textsf C}) \leqlm \emph{\AlgMem}_\KK(\textup{\textsf C})$ for any class of polynomials \textup{\textsf C} containing single variables and closed under multiplication with variables.
\end{theorem}

With this we can complete the proof of our main theorem.

\begin{proof}[Proof of \Cref{thm:main_result}]
The complexity-theoretic upper bounds have been established in \Cref{thm:algmem_expspace} and \Cref{thm:algmem_homog}. The lower bounds now follow from the hardness of $\IdealMem_\KK$ and its homogeneous variant (\Cref{thm:idealmem_expspace}) in conjunction with the reduction in \Cref{thm:reduction_ideal_algebra}.
\end{proof}

As mentioned in the beginning of this section, the complexity lower bounds for ideal membership were achieved by studying the word problem for commutative semigroups. Let $\mathcal{R} = \{(x^{\alpha_i},x^{\beta_i})\}_{i=1}^s$ be a set of pairs of monomials, called \emph{replacement rules}. Let $\equiv_{\mathcal R}$ be the congruence relation on $\Mon(\x)$ generated by
\[
x^\gamma x^\alpha\equiv_{\mathcal{R}} x^\gamma x^\beta \qquad (x^\alpha,x^\beta) \in \mathcal{R},\ x^\gamma \in \Mon(\x).
\]
Concretely, two monomials $m,m'$ are equivalent if one can be turned into the other by repeatedly replacing a factor $x^\alpha \mid m$ by the corresponding $x^\beta$ or by replacing $x^\beta \mid m$ by $x^\alpha$.

\begin{table}[H]
\centering
\problemBox[\textwidth]{Word problem for Commutative Semigroups, \CSG}
{Replacement rules $\mathcal{R} = \{(x^{\alpha_i}, x^{\beta_i})\}_{i=1}^s$, \newline
Monomials $m,m' \in \Mon(\x)$}
{Is $m \equiv_{\mathcal{R}} m'$?}
\end{table}
Mayr \& Meyer \cite{Mayr1982} prove \EXPSPACE-hardness of \CSG and provide a reduction $\CSG \leqlm \IdealMem_\KK(\textsf{Binom})$ via the map
\[
(\mathcal{R},m,m')\;\mapsto\;(x^{\alpha_1}-x^{\beta_1},\dots,x^{\alpha_s}-x^{\beta_s};m-m').
\]
More precisely, the minimal number of terms in a representation for this ideal membership problem equals the minimal number of replacement steps necessary in showing the congruence $m \equiv_{\mathcal R} m'$. Using the known double-exponential examples from \cite{Mayr1982}, we can prove worst-case certification degree lower bounds for subalgebra membership. For an accessible exposition of the double-exponential lower bounds compare \cite{Bayer1988} or the author's Master's thesis \cite{kayser2022groebnercomplexity}.

\begin{theorem}
For given $k$ there exist polynomials $f_1,\dots,f_s,g \in \KK[x_1,\dots,x_n]$ of degree $\leq 6$, $n = \Landau(k)$, $s = \Landau(k)$, such that $g \in \KK[f_1,\dots,f_s]$, but every certificate $p \in \KK[t_1,\dots,t_s]$ has degree at least $2^{2^n}$ and at least $2^{2^n}$ terms. These $f_i$ and $g$ can be chosen as binomials and monomials.
\end{theorem}
\begin{proof}
The \CSG instances from \cite{Mayr1982} lead to ideal membership of $\Landau(k)$ binomials in $\Landau(k)$ variables of degree $\leq 5$ (to obtain this exact degree bound, use the slightly improved degree bound from \cite{Bayer1988}) such that every representation has degree and number of terms bounded below by $2^{2^n}$. Applying the reduction from \Cref{thm:reduction_ideal_algebra}, we obtain an instance of $\AlgMem_\KK$, whose smallest certificate $p$ has the same number of terms and greater total degree by \Cref{cor:repdeg_to_certdeg}.
\end{proof}

\begin{remark}
Following \cite{Bayer1988}, this construction can be easily adapted to prove lower bounds of the form $d^{2^k}$ at the expense of increasing the algebra generator degree to $d+ \Landau(1)$. This asymptotically matches the double-exponential upper bounds from \Cref{cor:algmem_degree_bound}. It would be interesting to get more refined lower bounds in terms of other properties of the algebra. For example, what could analogous statements of \Cref{thm:hermann_bound} and \Cref{thm:gb_dim_mayr_ritscher} be?
\end{remark}

While the worst-case examples for \CSG are too complicated to recall here, we can give more details in the homogeneous case. Here the certification degree can not blow up, however, we can still produce $\AlgMem_\KK(\Homog)$ instances that require exponentially many terms in their certificates.

\begin{theorem}
There exist a subalgebra  $A = \KK[f_1,\dots,f_{5n}] \subseteq \KK[x_1,\dots,x_{3n+3}]$ generated by single variables and homogeneous binomials of degree $\leq 3$, and a binomial $g$ of degree $n+2$ with the following property:
$g \in A$, but every polynomial $p \in \KK[\t]$ with $g=p(f_1,\dots,f_{5n})$ has at least $2^{n+1}$ terms.
\end{theorem}
\begin{proof}
The construction is based on the binary-counting LBA from \Cref{exa:binaryLBA}. In \cite[Theorem 17]{Mayr1997} Mayr describes how to reduce \LBAhalt to commutative semigroups/ideals, proving \PSPACE-hardness of $\IdealMem_\KK(\Homog)$. Applying our reduction (\Cref{thm:reduction_ideal_algebra}), we obtain a binary-counting subalgebra.
The configuration $(q,i,b_0,\dots,b_{n+1})$ of the LBA is represented by the monomial $qh_i x_{0,{\triangleright}} x_{1,b_1} \dotsm x_{n,b_n}x_{n+1,{\triangleleft}}$ and the replacement rules/binomial generators mimic the transition function $\delta$.

We can simplity the resulting algebra significantly: We can remove redundant transitions (for example, we don't actually need ${\triangleright},{\triangleleft}$ since the head position variable knows the current position), and the new variable $t$ introduced in \Cref{thm:ideal_to_algebra} can also be omitted, as the combined variables $qh_i$ have a similar effect. This leads to the following subalgebra:
\begin{align*}
\mathcal{T} = \{q_0,q_1\} &\mathbin{\dot\cup} \{h_0,\dots,h_n\}, \quad \x = \{x_{1,\NULL},x_{1,\EINS},\dots,x_{n,\NULL},x_{n,\EINS}\}, \\
\mathcal{R} &\coloneqq \set{q_0h_ix_{i,\NULL} - q_1h_{i-1}x_{i,\EINS} | 1\leq i \leq n} \\
        &{}\; \cup \set{q_0h_ix_{i,\EINS} - q_0h_{i+1}x_{i,\NULL} | 1\leq i \leq n-1} \\
        &{}\; \cup \set{q_1h_ix_{i,\NULL} - q_1h_{i-1}x_{i,\NULL} | 1\leq i \leq n} \\
        &{}\; \cup \set{q_1h_0 - q_0h_1}, \\
  A &\coloneqq \KK[f_1,\dots,f_{5n}] \coloneqq \KK[\mathcal{T} \cup \x], \\
  g &\coloneqq q_0h_1x_{1,\NULL}\dotsm x_{n,\NULL} - q_0h_nx_{1,\NULL}\dotsm x_{n-1,\NULL}x_{n,\EINS}.
\end{align*}
Then $g \in A$ by construction, since the binary counter will go from $(q_0,1,\triangleright\NULL \dots \NULL \triangleleft)$ to $(q_0,1,\triangleright\EINS \dots \EINS \triangleleft)$ and then walk to the right, erasing the $\EINS$'s until it reaches the configuration $(q_0,n,\triangleright\NULL \dots \NULL \EINS \triangleleft)$. On the way it writes every number $0,\dots,2^{n-1}$ on the tape, taking at least $2$ steps each time, for a total of $\geq 2^{n+1}$ steps. By \Cref{cor:repdeg_to_certdeg} any certificate for $g \in A$ must essentially contain this derivation, hence has at least $2^{n+1}$ terms.
\end{proof}

An implementation of the homogeneous binary-counting subalgebra in \texttt{Macaulay2} can be found at \href{https://mathrepo.mis.mpg.de/ComplexityOfSubalgebras}{\nolinkurl{mathrepo.mis.mpg.de/ComplexityOfSubalgebras}}.

\section{Monomial subalgebras and SAGBI bases}\label{sec:monomials}

In this final section, we consider the complexity of monomial algebra membership and consider some questions related to SAGBI bases.
Monomial subalgebras $A$ are $\NN^n$-graded in the sense that
\[
A = \bigoplus_{\substack{\alpha \in \NN^n\\ x^\alpha \in A}} \KK x^\alpha \subseteq \KK[\x].
\]
This has the useful consequence that a polynomial $\sum_{\alpha} c_\alpha x^\alpha$ is in $A$ if and only if every monomial $x^\alpha$, $c_\alpha\neq 0$, is in $A$. Furthermore, monomial algebras are related to linear programming and linear Diophantine equations \cite[Remark 1.9]{Robbiano1990}, as
\begin{equation}\label{eq:diophantine}
  x^\beta \in K[x^{\alpha_1},\dots,x^{\alpha_s}] \quad \Longleftrightarrow\quad \exists c \in \NN^s\text{ s.t.\ } \beta = \sum_{i=1}^s c_i\alpha_s.
\end{equation}

\begin{theorem}\label{thm:Monomial_gen}
For any $\KK$ the problem $\emph{\AlgMem}_{\KK}(\Mon)$ is \NP-complete. This is true even when restricting to square-free monomials.
\end{theorem}
\begin{proof}
For \NP membership one immediately reduces to the case where the polynomial to test $f$ is a monomial due to the $\NN^n$-grading. In \Cref{eq:diophantine} we have $c_j \leq \max_i \beta_i$, so the bit length of $c_j$ is bounded by the bit length of $\beta$. Hence, non-deterministically guessing $c$ yields a \NP-algorithm.

For \NP-hardness, one can reduce from the \NP-complete problem  \OneInThreeSAT (\Cref{sec:complexity}). Indeed, given sets $S_1,\dots,S_n \subseteq \{1,\dots,s\}$, then let $\alpha_1,\dots,\alpha_s \in \{0,1\}^n$ be the integer vectors with $(\alpha_i)_j = 1$ when $i \in S_j$. Set $\beta = (1,\dots,1) \in \NN^n$. Then \Cref{eq:diophantine} encodes exactly the question for \OneInThreeSAT, a solution corresponding to $T = \set{i | c_i = 1}$. In this construction, all monomials $x^{\alpha_i},x^\beta$ are square-free.
\end{proof}

We see that $\AlgMem_\KK(\Mon)$ is \NP-complete even for polynomials of degree $\leq n$. The same is true if we instead bound the number of variables -- if the exponents are encoded in binary.

\begin{theorem}\label{thm:Monomial_bounded}
The problem $\emph{\AlgMem}_\KK(\Mon(x_1,\dots,x_n))$ for \emph{fixed} $n\geq 1$ is \NP-complete for binary exponent encoding and in \TCO with unary encoding.
\end{theorem}

Here $\TCO \subsetneq \DP$ is a low uniform circuit complexity class. This model of computation solves a decision problem on input $b_1\dots b_n$ by evaluating a $n$-ary boolean function $f_{C_n}(b_1,\dots,b_n)$. The function is described by a Boolean circuit $C_n$ of size polynomial in $n$ with a bounded number of layers. The gates of the circuits compute the boolean functions $\texttt{not}$ and $\texttt{and},\texttt{or},\texttt{maj}$ of any arity $\{0,1\}^n\to\{0,1\}$, where the majority gate is the function
\[
\texttt{maj}(b_1,\dots,b_n) = \begin{cases}
    1 & \text{if } \#\set{i|b_i=1} \geq \frac{n}{2} \\
    0 & \text{otherwise.}
\end{cases}
\]
For more information on $\TCO$ and circuit complexity classes, see \cite[Chapter 4]{Vollmer1999}.

\begin{proof}
Encoding the exponents as binary, the unary case $n=1$ is a direct translation of the \SubsetSum problem, which is \NP-hard.

On the other hand, if the monomials are encoded in unary, then a generating-function approach as in \cite{Kane2017}  provides a family of circuits in \TCO deciding $\AlgMem_K(\Mon(x_1,\dots,x_n))$. In more detail, Kane describes $\TCO$ circuits deciding the unary vector-valued subset sum problem, which corresponds to solutions of \Cref{eq:diophantine} with $c_i \in \{0,1\}$. We can reduce the case of arbitrary $c_i \in \NN$ to the $\{0,1\}$ case by replacing every generator $x^{\alpha_i}$ by $x^{\alpha_i},x^{2\alpha_i},x^{4\alpha_i},\dots,x^{2^\kappa \alpha_i}$, $\kappa = \lfloor \log_2 |\beta| \rfloor$. Since the numbers are encoded in unary, the input size is proportional to $C=|\beta|+\sum_{i=1}^s |\alpha_i|$, and this larger generating set has size $\Landau(C^2)$.

Alternatively, one can apply a \TCO-variant of Courcelle’s theorem to obtain \TCO-membership, compare \cite[Theorem 13]{Elberfeld2012}.
\end{proof}

\begin{remark}
The \NP-hardness results from \Cref{thm:Monomial_gen} and \ref{thm:Monomial_bounded} are in stark contrast to the analogous results for monomial \emph{ideals}: Monomial ideal membership is computationally trivial, as $x^\beta \in \langle x^{\alpha_1},\dots,x^{\alpha_s} \rangle$ if and only if component-wise $\beta \geq \alpha_i$ for some $i$.
\end{remark}

\subsection{Some remarks on SAGBI bases} \label{sec:SAGBIconjecture}

We return to a general subalgebra $A = \KK[f_1,\dots,f_s] \subseteq \KK[\x]$ equipped with a monomial order $\prec$. In \Cref{sec:ideals} we defined the concept of SAGBI bases $S \subseteq A$ as subsets with $\KK[\set{\ini_\prec f | f \in S}] = \ini_\prec A$.
SAGBI stands for \enquote{Subalgebra Analog to Gröbner Bases for Ideals} \cite{Robbiano1990} and they are used in practice for deciding subalgebra membership \cite{Burr2024}. Much of the theory of Gröbner bases is paralleled for SAGBI bases, such as the \emph{subduction algorithm} deciding subalgebra membership, reminiscent of the division algorithm for Gröbner bases \cite[Section 6.6]{Kreuzer2005-rb}.

The previous results in this section provide a modest complexity lower bound of deciding subalgebra membership in the presence of SAGBI bases. The author is not aware of an analogous result for ideal membership given a Gröbner basis.

\begin{corollary}
The problem $\emph{\AlgMem}_K$ is \NP-hard, even if the input polynomials $f_1,\dots,f_s$ form a finite SAGBI basis of $\KK[f_1,\dots,f_s]$.
\end{corollary}
\begin{proof}
Subalgebras generated by a finite set $S$ of monomials have $S$ as a SAGBI basis. Hence \Cref{thm:Monomial_gen} provides the \NP lower bound.
\end{proof}

A major difference to the ideal world is that not every finitely generated subalgebra has a \emph{finite} SAGBI basis.
\begin{example}
The subalgebra $\KK[x_1,x_1x_2-x_2^2,x_1x_2^2] \subseteq \KK[x_1,x_2]$ has the non-finitely generated initial algebra $\KK[\set{x_1x_2^k|k\geq 0}]$ for any monomial order with $x_1\succ x_2$ \cite[Example 4.11]{Robbiano1990}.
\end{example}

We therefore propose the study of the following two decision problems related to the initial algebra.

\begin{table}[H]
\centering
\problemBox[0.48\textwidth]{$\problemText{SAGBIfinite}_\KK(\textsf{C})$, $\problemText{SAGBIfinite}$}
{$f_1,\dots,f_s \in \textsf{C}$ (or $\KK[\x]$)} 
{Does $\KK[f_1,\dots,f_s]$ have a finite SAGBI basis?}
\problemBox[0.48\textwidth]{$\IniAlgMem_\KK(\textsf{C})$, $\IniAlgMem_\KK$}
{$f_1,\dots,f_s \in \textsf{C}$ (or $\KK[\x]$);\newline $x^\alpha \in \Mon(\x)$}
{Is $x^\alpha \in \ini_\prec \KK[f_1,\dots,f_s]$?}
\end{table}

Robbiano \& Sweedler showed that an algorithm somewhat analogous to Buchberger's algorithm for Gröbner bases can be used to enumerate a SAGBI basis, which will terminate if (and only if) $A$ has a \emph{finite} SAGBI basis. A procedure that returns \texttt{yes} if the output is yes, but never terminates if the output is no, is called a \emph{semi-algorithm}, and a problem is \emph{semi-decidable} or \emph{recursively enumerable} if there is a semi-algorithm \enquote{solving} it.

\begin{theorem}[Robbiano \& Sweedler \cite{Robbiano1990}]
$\emph{\IniAlgMem}_\KK$ and $\emph{\problemText{SAGBIfinite}}_\KK$ are semi-decidable (over a computable field).
\end{theorem}

We are not aware of any general better complexity bounds, or even just if these problems are computable at all (though a negative answer would be quite surprising). Future work will provide a more detailed study of the structure and complexity of (infinitely generated) initial algebras.

\section*{Acknowledgments}

I would like to thank my advisor, Simon Telen, as well as Markus Bläser, Peter Bürgisser, Florian Chudigiewitsch, and Fulvio Gesmundo for helpful discussions, in particular Florian for suggesting Courcelle's theorem for \Cref{thm:Monomial_bounded}. I am also indebted to the anonymous reviewers, whose reports helped to improve the presentation and significantly simplify the treatment of \Cref{sec:lower_bounds}. My interest in the complexity of ideals and subalgebras originated from my Master's thesis \cite{kayser2022groebnercomplexity} supervised by Heribert Vollmer \& Sabrina Gaube at Leibniz University Hannover.

\printbibliography

\subsection*{Author's address:}

\noindent Leonie Kayser, MPI-MiS Leipzig \hfill \href{mailto:leo.kayser@mis.mpg.de}{\tt leo.kayser@mis.mpg.de}

\end{document}